\documentclass[12pt]{article}
\usepackage[cp1251]{inputenc}
\usepackage{amssymb,amsmath}
\usepackage[dvips]{graphicx,color}
\usepackage{mathrsfs}
\usepackage{array, amsfonts, mathrsfs}
\usepackage{amssymb}
\usepackage{amsthm}
\usepackage{graphics, graphicx}
\usepackage{epstopdf}
\usepackage[colorlinks=true]{hyperref}
\usepackage{float}

\newtheorem{Corollary}{Corollary}
\newtheorem{Proposition}{Proposition}
\newtheorem{Definition}{Definition}
\newtheorem{Hypothesis}{Hypothesis}

\title{On the M.Kac problem with augmented data}
\author{M.I.Belishev\thanks {St.Petersburg Department of Steklov Mathematical
        Institute, St.Petersburg, Russia, e-mail:
        belishev@pdmi.ras.ru},\,\,\,A.F.Vakulenko\thanks {St.Petersburg Department of Steklov Mathematical
        Institute, St.Petersburg, Russia, e-mail:
        vak@pdmi.ras.ru}}
\date{}

\begin{document}

\maketitle

\begin{abstract}
Let ${\Omega}$ be a bounded plane domain. As is known, the spectrum
$0<\lambda_1<\lambda_2\leqslant\dots$ of its Dirichlet Laplacian
$L=-\Delta{\upharpoonright}[H^2({\Omega})\cap H^1_0({\Omega})]$ does not determine
${\Omega}$ (up to isometry). By this, a reasonable version of the
M.Kac problem is to augment the spectrum with relevant data that
provide the determination.

To give the spectrum is to represent $L$ in the form $\tilde
L=\Phi L\Phi^*={\rm diag\,}\{\lambda_1,\lambda_2,\dots\}$ in the
space ${\bf l}_2$, where $\Phi:L_2({\Omega})\to{\bf l}_2$ is the
Fourier transform. Let ${\mathscr K}=\{h\in
L_2({\Omega})\,|\,\,\Delta h=0\,\,{\rm in}\,\,{\Omega}\}$ be the
harmonic function subspace, $\tilde{\mathscr K}=\Phi{\mathscr
K}\subset{\bf l}_2$. We show that, in a generic case, the pair
$\tilde L,\tilde {\mathscr K}$ determines ${\Omega}$ up to
isometry, what holds not only for the plain domains (drums) but
for the compact Riemannian manifolds of arbitrary dimension,
metric, and topology. Thus, the subspace $\tilde{\mathscr
K}\subset{\bf l}_2$ augments the spectrum, making the problem
uniquely solvable.
\end{abstract}

\rightline{\bf Dedicated to the jubilee of Nina Nikolaevna
Ural'tseva}

\section{Introduction}\label{Sec Introduction}

\subsubsection*{About the paper}
{\noindent$\bullet$\,\,\,} Let ${\Omega}$ be a bounded plane domain (drum). The M.Kac
problem is to determine the shape of ${\Omega}$ from the spectrum
$0<\lambda_1<\lambda_2\leqslant\dots$ of its Dirichlet Laplacian
$L=-\Delta{\upharpoonright}[H^2({\Omega})\cap H^1_0({\Omega})]$. The well-known
fact is that a unique (up to isometry) determination is
impossible: there exist isospectral but not isometric domains
\cite{Buser,Giraud}. By this, a reasonable version of the problem
is to augment the spectrum with relevant data that provide the
determination.

We provide such data. Moreover, in a generic case, the augmented
data enable to determine not only a plane domain but a Riemannian
manifold of arbitrary dimension, metric, and topology.
\smallskip

{\noindent$\bullet$\,\,\,} A construction that is used for the determination was
originated in \cite{B Kac_Prob}. In \cite{B JOT} it was recognized
as a {\it wave spectrum} of the symmetric semi-bounded operator.
The paper \cite{BSim_FAN} clarifies its metric and topological
background. In \cite{BSim_MS} the wave spectrum is used for
determination of the Riemannian manifold from its spectral data.
In the given paper, we adapt the approach \cite{BSim_MS} to deal
with M.Kac problem.

\subsubsection*{Content and results}
{\noindent$\bullet$\,\,\,} The wave spectrum is a set of atoms of a relevant lattice of
the subspaces in a Hilbert space. Its construction is rather
complicated technically. We provide a brief presentation of all
elements and steps of the construction, beginning with the basic
lattice theory notions \cite{Birhoff}. The main element is the
correspondence (isomorphism) between a metric lattice in ${\Omega}$
and a relevant Hilbert (sub)lattice in $L_2({\Omega})$.

The original data of the M.Kac problem is the spectrum $\sigma(L):
\,0<\lambda_1<\lambda_2\leqslant\dots$ of the Dirichlet Laplacian
$L=-\Delta{\upharpoonright}[H^2({\Omega})\cup H^1_0({\Omega})]$.
The Fourier transform $\Phi:L_2({\Omega})\to{\bf l}_2$ associated
with the eigen-basis of $L$, maps it to $\tilde L=\Phi
L\Phi^*={\rm diag\,}\{\lambda_1,\lambda_2,\dots\}$. Let ${\mathscr
K}=\{h\in L_2({\Omega})\,|\,\,\Delta h=0\,\,{\rm
into}\,\,{\Omega}\}$ be the harmonic function subspace, and let
$\tilde{\mathscr K}=\Phi{\mathscr K}\subset{\bf l}_2$ be its
spectral representation. The subspace $\tilde {\mathscr K}$ can be
given by a numerical array (an infinite matrix $\varkappa$) that
generalizes the relevant data in one-dimensional inverse problems
(see Comments at the end of the paper).

Our main result is that, in a generic case, the pair $\tilde
L,\tilde{\mathscr K}$ (equivalently, its numerical version
$\sigma(L),\varkappa$) determines ${\Omega}$ up to isometry, what
holds not only for the plain domains but for a generic class of
the compact Riemannian manifolds of arbitrary dimension, metric,
and topology. Thus, the subspace $\tilde{\mathscr K}\subset{\bf
l}_2$ augments the spectrum, gaining a uniqueness (up to isometry)
in the M.Kac problem.
\smallskip

{\noindent$\bullet$\,\,\,} At the end of the paper, it is hypothesized that for
determination of ${\Omega}$ it suffices to have two spectra:
$\sigma(L)$ and $\sigma(L_M)$, where $L_M$ is the (self-adjoint)
extension by M.G.Krein of the minimal Laplacian $L_0=-\Delta{\upharpoonright}
H^2_0({\Omega})$ (the so-called {\it soft extension}).
\smallskip

{\noindent$\bullet$\,\,\,} The authors are grateful to and S.A.Simonov for the useful
discussions on the subject of the paper.

\section{Geometry and Lattices}\label{Sec Lattices}

\subsubsection*{Manifold}

{\noindent$\bullet$\,\,\,} Let ${\Omega}$ be a smooth \footnote{everywhere {\it smooth}
means $C^\infty$-smooth} compact Riemannian manifold with a smooth
metric tensor $g$ and the smooth boundary
${\Gamma}:=\partial{\Omega}$, ${\rm dim\,}{\Omega}=n\geqslant 2$. Let
$\rm d$ be the distance in ${\Omega}$. The distance determines the
tensor, so that to give the metric space $({\Omega},\rm d)$ is to
give ${\Omega}$ as a Riemannian manifold.

For a set $A\subset{\Omega}$, let
$A^r:=\{x\in{\Omega}\,|\,\,x'\in{\Omega}\,|\,\,{\rm d\,}(x,x')<r\}$ be
its metric neighborhood of the radius $r>0$; so that $x^r$ is an
open ball of the radius $r$, centered at $x\in{\Omega}$. A boundary of
$A\subset{\Omega}$ is $\partial A:=\overline{A}\setminus{\rm
int\,}A$.

The sets
\begin{equation}\label{Eq Gamma^t}
{\Gamma}^t\,:=\,\{x\in{\Omega}\,|\,\,{\rm d\,}(x,{\Gamma})<t\},\quad t>0
\end{equation}
are the near-boundary subdomains of ${\Omega}$. By compactness of ${\Omega}$
we have $T_*:=\inf\,\{t>0\,|\,\,{\Gamma}^t={\Omega}\}<\infty$.
\smallskip

{\noindent$\bullet$\,\,\,} Let $\mu$ be the Riemannian measure (volume) on ${\Omega}$.
There are its well-known properties:

\noindent{\bf 1.}\,\,\,For any $A\subset{\Omega}$, the measure
$\mu(A^t)$ is a continuous function of $t>0$ satisfying
$\mu(\overline A)=\mu(A^{+0})$. As a consequence, for the metric
neighborhood boundary one has $\partial
A^t=\{x\in{\Omega}\,|\,\,{\rm d\,}(x,A)=t\}$ and $\mu(\partial
A^t)=0$,\,\,\,\,$t>0$.

\noindent{\bf 2.}\,\,\,For any ball, $\mu(x^r)>0$ and
$\mu(x^{+0})=0$ holds for all $x\in{\Omega}$.

These properties are relevant for the wave spectrum construction
in a more general case of a metric space with measure
\cite{BSim_MS}.

\subsubsection*{Lattices}

The minimal necessary information from lattice theory is given
here. For detail, see \cite{Birhoff}.
\smallskip

{\noindent$\bullet$\,\,\,} A lattice is a partially ordered (by an order $\leqslant$)
set ${\mathfrak L}$ endowed with two operations $p\wedge q:=\inf\{p,q\}$ and
$p\vee q:=\sup\{p,q\}$. We assume that ${\mathfrak L}$ contains the extreme
(lowest and greatest) elements $0$ and $1$.

There is a canonical ({\it order}) topology on ${\mathfrak L}$ determined by
the {\it order convergence}: for a net $p_\alpha$, the convergence
$p_\alpha\overset{o}\to p$ means that $p=
\underset{\alpha}\vee\underset{\beta>\alpha}\wedge
p_\beta=\underset{\alpha}\wedge\underset{\beta>\alpha}\vee
p_\beta$ holds.

We say ${\mathfrak L}$ to be a lattice with {\it complement} if for any
$p\in{\mathfrak L}$ there is a unique $p^\bot\in{\mathfrak L}$ such that $p\wedge
p^\bot=0$ and $p\vee p^\bot=1$ holds. In this case, the lattice
possesses the complement operation $\bot:p\mapsto p^\bot$.

A map $i:{\mathfrak L}\to{\mathfrak L}$ is {\it isotonic} if it prevents the order,
i.e., $p\leqslant q$ implies $i(p)\leqslant i(q)$,\,\,\,$p,q\in{\mathfrak L}$. An {\it
isotony} $I=\{i^t\}_{t\geqslant 0}$,\,\,\,$i^0:=\rm id$ is a monotone
family of isotonic maps, so that $p\leqslant q$ and $s<t$ imply
$i^s(p)\leqslant i^t(q)$.
\smallskip

Let ${\mathfrak L}$ be a lattice with a complement $\bot$ and isotony $I$.
Let $A_\alpha\subset{\mathfrak L}$ be a family of subspaces. By $\sqcup
A_\alpha\subset\mathfrak L$ we denote a (sub)lattice of $\mathfrak
L$ that

\noindent(a)\,\,\, contains $A_\alpha$,

\noindent(b)\,\,\,is closed w.r.t. the order topology:
$\overline{\sqcup A_\alpha}=\sqcup A_\alpha$, and invariant w.r.t.
complement operation: $\bot\,\sqcup A_\alpha=\sqcup A_\alpha$,

\noindent(c)\,\,\, is invariant w.r.t. the isotony: $I\sqcup\!
A_\alpha\subset\sqcup A_\alpha$,

\noindent(d)\,\,\,is the smallest of the lattices having the
properties (a), (b), (c).

\noindent We say that the family $A_\alpha$ generates the lattice
$\sqcup A_\alpha$.
\smallskip

{\noindent$\bullet$\,\,\,} The set $\mathfrak F_\mathfrak L$ of the ${\mathfrak L}$-valued
functions of the variable $t\geqslant 0$ constitutes a lattice with the
point-wise order and operations
$$
f\leqslant g \Leftrightarrow f(t)\leqslant g(t);\,\,\, (f\wedge
g)(t):=f(t)\wedge g(t);\,\,\, (f\vee g)(t):=f(t)\vee g(t),\quad
t\geqslant 0,
$$
the extreme elements $0(t)=0_{\mathfrak L}$,\, $1(t)=1_{\mathfrak L}$,\quad $t\geqslant 0$,
and the point-wise convergence $f\overset{o}\to g\Leftrightarrow
f(t)\overset{o}\to g(t),\quad t\geqslant 0$.

The lattice ${\mathfrak F}_{\mathfrak L}$ contains a set
$$
I{\mathfrak L}\,:=\,\{Ip\,|\,\,(Ip)(t)=i^t(p),\,\,p\in{\mathfrak L},\,\,\,t\geqslant 0\}
$$
and its topological closure $\overline{I{\mathfrak L}}$.
\smallskip

{\noindent$\bullet$\,\,\,} Let a partially ordered set $\mathscr P$ contain the lowest
element $0$. An element $a\in\mathscr P$ is called an {\it atom}
if $0\not=p\leqslant a$ implies $p=a$. By ${\rm At\,}\mathscr P$ we
denote the set of atoms.

The sets ${\rm At\,}\overline{I{\mathfrak L}}$ play the key role in our
further considerations.

\subsubsection*{Metric lattices}

From this point until section \ref{Sec Dynamics}, the narrative is
in the nature of a brief recounting of facts and results from
\cite{BSim_FAN} and \cite{BSim_MS}. As before, ${\Omega}$ is a
Riemannian manifold.
\smallskip

{\noindent$\bullet$\,\,\,} The open set family ${\mathfrak O}:=\{G\subset 2^{\Omega}\,|\,\,G={\rm
int\,}G\}$ is a lattice with the order $\subseteq$, extreme
elements $0=\emptyset$ and $1={\Omega}$ and the operations
$\wedge=\cap,\,\,\vee=\cup$. It is endowed with the {\it metric
isotony}
$$
I=\{i^t\}_{t\geqslant 0}\,,\qquad i^t(G):=G^t.
$$
Evidently, the lattice ${\mathfrak O}$ does not contain atoms. In the mean
time, the set $\overline{I{\mathfrak O}}\subset{\mathfrak F}_{\mathfrak O}$ of the ${\mathfrak O}$-valued
functions does possess the atoms, which are characterized as
follows:
\begin{equation*}\label{Eq geom atoms single}
{\rm At\,}\overline{I{\mathfrak O}}=\{a_x\,|\,\,x\in{\Omega}\},\qquad
a_x(t)\,:=\,(Ix)(t)=x^t,\quad t\geqslant 0,
\end{equation*}
so that the correspondence ${\Omega}\ni x\leftrightarrow a_x\in {\rm
At\,}\overline{I{\mathfrak O}}$ is a bijection. Moreover, the function ${\rm
d'}: {\rm At\,}\overline{I{\mathfrak O}}\times{\rm
At\,}\overline{I{\mathfrak O}}\to\overline{\mathbb R_+},$
$$
{\rm d'}(a,b)\,:=\,2\,\inf\,\{t> 0\,|\,\,a(t)\cap
b(t)\not=\emptyset\}
$$
is a metric, whereas the equality ${\rm d'}(a_x,a_y)={\rm
d\,}(x,y)$ holds. As a consequence, the map $\beta: x\mapsto a_x$
turns out to be an isometry from ${\Omega}$ onto ${\rm
At\,}\overline{I{\mathfrak O}}$.
\medskip

{\noindent$\bullet$\,\,\,} For our future goals, a `drawback' of the lattice ${\mathfrak O}$ is
that it contains the sets with $\mu(\partial G)\not=0$ and does
not admit the complement $\bot: G\mapsto G^\bot={\Omega}\setminus
G$. Removing the first one, we make use of a smaller lattice
$$
{\mathfrak O}':=\{G\in{\mathfrak O}\,|\,\,\mu(\partial G)=0\}\,.
$$
As is easy to check, the equality ${\rm At\,}\overline{I{\mathfrak O}'}={\rm
At\,}\overline{I{\mathfrak O}}$ holds. Hence, denoting ${\Omega}':={\rm
At\,}\overline{I{\mathfrak O}'}$, we conclude that the metric space
$({\Omega}',\rm d')$ is an isometric copy ({\it wave model} - see
\cite{BSim_FAN}) of the original $({\Omega},\rm d)$.
\medskip

{\noindent$\bullet$\,\,\,} To remove the second drawback needs more work of technical
character.

Recall that $A\Delta B:=[A\setminus B]\cup[B\setminus A]$. We say
$A$ and $B$ to be equivalent and write $A\sim B$ if $\mu(A\Delta
B)=0$ holds. By $\langle A\rangle$ we denote the equivalence class of $A$.

If the open sets $G$ and $G'$ are equivalent then
$\overline{G}=\overline{G'}$ holds and implies
$\overline{G}^{\,\,t}=\overline{G'}^{\,\,t},\,\,\,t>0$. Each
equivalence class $\langle G\rangle$ consists of not only open sets but
contains a canonical representative $\dot G={\rm
int\,}\overline{\dot G}\in\langle G\rangle$. Due to that, the metric
isotony is extended from sets to classes by $\langle G\rangle^t:=\langle
G^t\rangle=\langle {\dot G}^{\,\,t}\rangle$, so that the map $\langle
G\rangle\mapsto\langle G\rangle^t$ is well defined (does not depend on the
representative $G\in\langle G\rangle$) and isotonic.

Eventually, we arrive at the lattice of the classes
$$
{\mathfrak L}^{\rm m}_{\Omega}\,:=\,\{\langle G\rangle\,|\,\,G\in{\mathfrak O}'\}
$$
with the order $\langle G\rangle\leqslant\langle G'\rangle \Leftrightarrow
\mu(G'\setminus G)\geqslant 0$, the extreme elements $0=\langle
\emptyset\rangle$ and $1=\langle {\Omega}\rangle$, the operations
$$
\langle G\rangle\cap\langle G'\rangle:=\langle G\cap G'\rangle;\quad \langle G\rangle\cup\langle
G'\rangle:=\langle G\cup G'\rangle;\quad \langle G\rangle^\bot:= \langle G^\bot\rangle;
$$
and the isotony
$$
I_\star=\{i_\star^t\}_{t\geqslant 0}\,,\quad i^t_\star:\langle
G\rangle\,\mapsto\,\langle G\rangle^t .:=\langle
$$
The index $\rm m$ in the notation  ${\mathfrak L}^{\rm m}_{\Omega}$ indicates
the metric nature and connection with the metric isotony.
\smallskip

{\noindent$\bullet$\,\,\,} The lattice ${\mathfrak F}_{{\mathfrak L}_{\Omega}^{\rm m}}$ of the
${\mathfrak L}_{\Omega}^{\rm m}$-valued functions contains the closed set
$\overline{I_\star{\mathfrak L}_{\Omega}^{\rm m}}$, for which one has
\begin{equation*}\label{Eq geom atoms}
{\rm At\,}\overline{I_\star{\mathfrak L}_{\Omega}^{\rm
m}}=\{a^\star_x\,|\,\,x\in{\Omega}\},\qquad a^\star_x(t)\,:=\,\langle
a_x(t)\rangle=\langle x^t\rangle,\quad t\geqslant 0,
\end{equation*}
so that there is a bijection ${\Omega}\ni x\leftrightarrow
a^\star_x\in {\rm At\,}\overline{I{\mathfrak L}_{\Omega}^{\rm m}}$. Moreover,
the function
$$
{\rm d}_\star(a^\star,b^\star)\,:=\,2\,\inf\,\{t>
0\,|\,\,a^\star(t)\cap b^\star(t)\not=\emptyset\}, \qquad
a^\star,b^\star\,\in\,{\rm At\,}\overline{I{\mathfrak L}_{\Omega}^{\rm m}}
$$
is a metric, whereas the equality ${\rm
d}_\star(a^\star_x,a^\star_y)={\rm d\,}(x,y)$ holds and the map
$\beta_\star: x\mapsto a^\star_x$ turns out to be an isometry from
${\Omega}$ onto ${\rm At\,}\overline{I{\mathfrak L}_{\Omega}^{\rm m}}$.
\smallskip

As the result, denoting ${\Omega}_\star:={\rm
At\,}\overline{I{\mathfrak L}_{\Omega}^{\rm m}}$, 
we have an isometry $\beta_\star:({\Omega}, \rm
d)\to({\Omega}_\star,{\rm d}_\star)$. The distance ${\rm d}_\star$
determines the whole Riemannian structure on ${\Omega}_\star$ that
turns the latter into an isometric copy (model) of the original
Riemannian manifold ${\Omega}$.

In what follows we refer on $({\Omega}_\star,{\rm d}_\star)$ as a
{\it metric wave model} of the manifold ${\Omega}$\,
\cite{BSim_FAN,BSim_MS}. The relevancy of such a name will become
clear later.


\subsubsection*{Hilbert lattices}

{\noindent$\bullet$\,\,\,} Let ${\mathfrak L}$ be the lattice of the (closed) subspaces of the
Hilbert space $\mathscr H=L_{2,\,\mu}({\Omega})$ endowed with the
operations
$$
\mathscr A\wedge\mathscr B=\mathscr A\cap\mathscr B,\,\,\,
\mathscr A\vee\mathscr B=\overline{\{a+b\,|\,\,a\in\mathscr
A,\,\,b\in\mathscr B\}},\,\,\,\mathscr A\mapsto\mathscr
A^\bot=\mathscr H\ominus\mathscr A,
$$
the extreme elements $0=\{0\}$ and $1=\mathscr H$. As can be shown, the
order topology coincides with the one determined by the strong
operator convergence of the corresponding (orthogonal)
projections: $\mathscr A_\alpha\overset{o}\to\mathscr A\,
\Leftrightarrow\,P_{\mathscr A_\alpha}\to P_\mathscr A$.

The subspaces of the form $\mathscr
H\langle{\omega}\rangle:=\{\chi_{\omega} y\,|\,\,y\in\mathscr
H\}$, where ${\omega}\subset{\Omega}$ is a measurable set and
$\chi_{\omega}$ is its indicator, are called {\it local}. For
${\omega}\sim{\omega}'$ one has $\mathscr
H\langle{\omega}\rangle=\mathscr H\langle{\omega}'\rangle$, so
that a local subspace  is determined by the equivalence class
$\langle{\omega}\rangle$. Conversely, $\mathscr
H\langle{\omega}\rangle=\mathscr H\langle{\omega}'\rangle$ leads
to ${\omega}\sim{\omega}'$.

The set
$$
{\mathfrak L}_{\Omega}^{\rm h}\,:=\,\{\mathscr H\langle G\rangle\,|\,\,\langle G\rangle \in
{\mathfrak L}_{\Omega}^{\rm m} \}
$$
is a topologically closed sublattice of ${\mathfrak L}$, which is invariant
w.r.t. the isotony
\begin{equation}\label{Eq Hilbert isotony}
I_*=\{i_*^t\}_{t\geqslant 0}\,,\qquad i_*^t:\mathscr H\langle G\rangle\mapsto\mathscr H\langle
G^t\rangle.
\end{equation}
The index $\rm h$ indicates a Hilbert space background of this
construction.
\begin{Proposition}\label{Prop 1}
\cite{BSim_MS}\,\,\,The map $\theta:\langle G\rangle\mapsto\mathscr H\langle G\rangle$ is
a lattice isomorphism from ${\mathfrak L}_\Omega^{\rm  m}$ onto
${\mathfrak L}_\Omega^{\rm h}$: it preserves all lattice operations and
intertwines the isotonies: $\theta I_\star=I_*\theta$.
\end{Proposition}

Denote ${\Omega}_*:={\rm At\,}\overline{I_*{\mathfrak L}_{\Omega}^{\rm h}}$. The
isomorphism of the lattices induces a canonical correspondence
between atoms
$$
{\Omega}_\star\,\ni\,a^\star_x\leftrightarrow
a^*_x\in{\Omega}_*:\qquad a^\star_x(t)=\langle x^t\rangle,\,\,\,
a^*_x(t)=\mathscr H\langle x^t\rangle,\quad t\geqslant 0,\,\,x\in{\Omega}.
$$
Introducing ${\rm d}_*(a^*,b^*):=2\inf\{t>0\,|\,\,a^*_x(t)\cap
b^*_x(t)\not=0_\mathscr H\}$, we get the isometric manifolds
\begin{equation}\label{Eq Copies}
({\Omega}_*,{\rm d}_*)\leftrightarrow({\Omega}_\star,{\rm
d}_\star)\leftrightarrow({\Omega}, \rm d)
\end{equation}
or, to say better, two isometric copies  $({\Omega}_*,{\rm d}_*)$
and $({\Omega}_\star,{\rm d}_\star)$ of the original $({\Omega}, \rm
d)$. The map $\beta_*:x\mapsto a^*_x$ is an isometry from ${\Omega}$
onto ${\Omega}_*$. In what follows we name $({\Omega}_*,{\rm d}_*)$ by
a {\it Hilbert wave model} of the manifold ${\Omega}$.

\subsubsection*{Simple manifolds}

{\noindent$\bullet$\,\,\,} Let us return to the metric lattices in
${\Omega}$. We say the family of the extending neighborhoods
${\gamma}:=\{{\Gamma}^t\}_{t>0}$ (see (\ref{Eq Gamma^t})) to be a
{\it boun\-dary nest}. By the properties of the Riemannian volume,
the equality $\mu(\partial {\Gamma}^t)=0$ holds, so that the
embedding ${\gamma}^{\rm
m}:=\{\langle{\Gamma}^t\rangle\}_{t>0}\subset\mathfrak L^{\rm
m}_{\Omega}$ occurs.

By $\mathfrak L^{\rm m}_{\Gamma}:=\sqcup{\gamma}^{\rm m}$ we
denote the (sub)lattice in $\mathfrak L^{\rm m}_{\Omega}$
generated by the boundary nest.
\begin{Definition}\label{Def 1}
The manifold ${\Omega}$ is called {\it simple} if the set
$\{G\,|\,\,\langle G \rangle\subset \mathfrak L^{\rm m}_{\Gamma}\}$ is a
base of the (metric) topology in ${\Omega}$.
\end{Definition}
By this definition, for any $x\in{\Omega}$ and (arbitrarily small)
$\varepsilon>0$ there is an open set $G\in\langle G\rangle\in \mathfrak
L^{\rm m}_{\Gamma}$ such that $x\in G$ and ${\rm
diam\,}G<\varepsilon$ holds. This easily follows to $ {\rm
At\,}\overline{I_\star\mathfrak L^{\rm m}_{\Gamma}}={\rm
At\,}\overline{I_\star\mathfrak L^{\rm m}_{\Omega}}={\Omega}_\star$.
\smallskip

An evident obstacle for ${\Omega}$ not to be simple is its
symmetries. For instance, if ${\Omega}=\{x\in\mathbb
R^n\,|\,\,|x|=R\}$ is a ball then, as is easy to verify, its
metric copy $({\Omega}_\star,\rm d_\star)$ is isometric to the
segment $[0,R]\subset \mathbb R$. In the mean time, on a
Riemannian manifold, a ball $x^r\subset {\Omega}$ of the radius less
than the injectivity radius at $x$ may have the trivial symmetry
group but its metric copy is also isometric to the segment
$[0,r]$. A triangle $T\subset\mathbb R^2$ with the pair-wise not
equal length of the sides is simple: its metric copy is isometric
to $T$ itself. The simplicity of ${\Omega}$ is a generic property:
it can be gained by the arbitrarily small smooth variations of the
boundary $\partial{\Omega}$.

\smallskip

{\noindent$\bullet$\,\,\,} Let us outline the further use of the copies (\ref{Eq
Copies}) in the determination of ${\Omega}$. The manifold ${\Omega}$
is assumed {\it simple}.
\smallskip

Let us be given the nest ${\gamma}^{\rm m}$ and the isotony
$I_\star$. Then we can determine the metric copy of ${\Omega}$ by
the scheme
\begin{equation*}
{\gamma}^{\rm
m},\,I_\star\,\,\overset{\sqcup}{\,\,\Rightarrow}\,\,\mathfrak
L^{\rm m}_{\Gamma}\,\,\Rightarrow\,\,{\rm
At\,}\overline{I_\star\mathfrak L^{\rm
m}_{\Gamma}}\,\,\Rightarrow\,\,({\Omega}_\star,\rm d_\star).
\end{equation*}

Denote $\gamma^{\rm h}:=
\theta \gamma^{\rm m}
=
\{  {\mathscr H}\langle
{\Gamma}^t\rangle\}_
{t\geqslant 0}\subset {\mathfrak L}^{\rm
h}_{\Omega}=\theta\mathfrak L^{\rm m}_{\Omega}$. The Hilbert copy can
be constructed just by reproducing the previous scheme:
\begin{equation*}\label{Eq scheme Hilbert}
{\gamma}^{\rm
h},\,I_*\,\,\overset{\sqcup}{\,\,\Rightarrow}\,\,\mathfrak L^{\rm
h}_{\Gamma}\,\,\Rightarrow\,\,{\rm At\,}\overline{I_*\mathfrak
L^{\rm h}_{\Gamma}}\,\,\Rightarrow\,\,({\Omega}_*,\rm d_*).
\end{equation*}

The Hilbert copy is of invariant nature in the following sense.
Let $U:{\mathscr H} \to \tilde {\mathscr H}$ be a unitary
operator. There are the objects in $\tilde{\mathscr H}$ dual to
ones in $\mathscr H$: $\tilde{\gamma}^{\rm h}=U{\gamma}^{\rm
h}=\{U\mathscr H\langle{\gamma}^t\rangle\}_{t\geqslant 0}$,
$\tilde I_*=UI_*U^*$, $\tilde{\mathfrak L}^{\rm h}_{\Gamma}=\sqcup
{\tilde{\gamma}}^{\rm h}$, and so on. By the use of them, one
constructs an isometric copy of the original manifold ${\Omega}$
as follows:
\begin{equation}\label{Eq scheme tilde}
\tilde{\gamma}^{\rm h},\,\tilde
I_*\,\,\overset{\sqcup}{\,\,\Rightarrow}\,\,\tilde{\mathfrak
L}^{\rm h}_{\Gamma}\,\,\Rightarrow\,\,{\rm At\,}\overline{\tilde
I_*\tilde{\mathfrak L}^{\rm
h}_{\Gamma}}\,\,\Rightarrow\,\,(\tilde{\Omega}_*,\tilde{\rm
d}_*)\overset{\rm isom}=({\Omega},\rm d).
\end{equation}
It is the scheme that is relevant for determination of
$({\Omega},\rm d)$: we'll be given a concrete space $\tilde
{\mathscr H}$ together with the nest ${\tilde{\gamma}}^{\rm h}$
and isotony $\tilde I_*$ in it that provides the staring point for
the procedure (\ref{Eq scheme tilde}).
\smallskip

We refer to $({\tilde{\Omega}}_*,\tilde{\rm d}_*)$ as a Hilbert wave
model of ${\Omega}$ in the $U$-representation.

\section{Dynamics}\label{Sec Dynamics}

\subsubsection*{Operators}
Recall that $\mathscr H=L_{2,\,\mu}({\Omega})$. Let $L_0:=-\Delta{\upharpoonright}
H^2_0({\Omega})=\overline{-\Delta{\upharpoonright} C^\infty_0({\Omega})}$ be the
{\it minimal} Beltrami-Laplace operator in $\mathscr H$. It is a
symmetric completely non-selfadjoint positive definite operator,
so that $(L_0y,y)\geqslant\nu\|y\|^2$ holds with a $\nu>0$ for
$y\in{\rm Dom\,} L_0$. Its defect indexes are $n^\pm=\infty$. Its
Friedrichs extension is $L:=-\Delta{\upharpoonright}[H^2({\Omega})\cap
H^1_0({\Omega})]$ satisfies $L=L^*$ and $(Ly,y)\geqslant\nu\|y\|^2$.
Its adjoint is the {\it maximal} Laplacian
$L_0^*=-\Delta{\upharpoonright}[{\rm Dom\,}L\dotplus {\mathscr K}]$, where
$$
{\mathscr K}:={\rm Ker\,}L_0^*=\{y\in\mathscr H\,|\,\,\Delta y=0\,\,{\rm
in}\,\,{\Omega}\setminus{\Gamma}\}
$$
is the harmonic function subspace, ${\rm dim\,}{\mathscr K}=n^\pm=\infty$.
So, we have an operator triple
\begin{equation}\label{Eq oper triple}
L_0\,\subset\,L\,\subset\,L_0^*.
\end{equation}

\subsubsection*{System $\alpha$}
{\noindent$\bullet$\,\,\,} With the operators one associates a dynamical system $\alpha$
of the form
\begin{align}
\label{Eq 1}& u_{tt}-\Delta u = 0  && {\rm in}\,\,\,({\Omega}\setminus{\Gamma})\times\mathbb R_+;\\
\label{Eq 2}& u|_{t=0}=u_t|_{t=0}=0 && {\rm in}\,\,{\Omega};\\
\label{Eq 3}& u = f && {\rm on}\,\,\,{\Gamma}\times \overline{\mathbb R_+};
\end{align}
where $t$ is a time; $\Delta$ is the Beltrami-Laplace operator in
${\Omega}$,  $f$ is a {\it boundary control}. By $u=u^f(x,t)$ we
denote the solution ({\it wave}). For smooth controls 
$f\in C^\infty({\Gamma}\times\overline{\mathbb R_+})$ vanishing near
$t=0$, the solution $u^f$ is unique and classical.

As is known, the cor\-res\-pondence $f\mapsto u^f(\cdot,t)$ is
continuous from $L_{2,\,{\rm d}{\Gamma}}({\Gamma}\times[0,t])$ to
$\mathscr H$ for all $t>0$, where ${\rm d}{\Gamma}$ is the surface element
induced by the volume $\mu$. Note that the class of the smooth
controls is dense in $L_{2,\,{\rm d}{\Gamma}}({\Gamma}\times[0,t])$
for any $t>0$.
\smallskip

{\noindent$\bullet$\,\,\,} The sets ${\mathscr U}^t:=\{u^f(\cdot,t)\,|\,\,f\in L_{2,\,{\rm
d}{\Gamma}}({\Gamma}\times[0,t])\}\,\,\,(t> 0)$ are called {\it
reachable}. A fact that plays the key role in our approach to
inverse problems, is as follows.
\begin{Proposition}\label{Prop Holmgren}
The equality
\begin{equation}\label{Eq Holmgren}
\overline{{\mathscr U}^t}\,=\,\mathscr H\langle {\Gamma}^t\rangle,\qquad t>0
\end{equation}
holds
\end{Proposition}
\noindent(see \cite{B Obzor IP 97}). It is based upon the
fundamental Holmgren-John-Tataru Theorem on uniqueness of
continuation of the hyperbolic equation solutions across a
noncharacteristic surface. This equality is interpreted as a {\it
local approximate boundary controllability of the system $\alpha$}
or as a completeness of waves in the domains that they fill up.
Note that, by the latter, due to the compactness of ${\Omega}$  we
have $\overline{{\mathscr U}^t}=\mathscr H$ for $t\geqslant T_*$.

\subsubsection*{Wave nest}

The extending family of the subspaces ${\gamma}^{\rm
w}:=\{\overline{{\mathscr U}^t}\}_{t\geqslant 0}\subset\mathscr H$
is called a {\it wave nest}. As a consequence of the equality
(\ref{Eq Holmgren}), we have the coincidence of the nests:
\begin{equation}\label{Eq gamma^w=gamma^h}
{\gamma}^{\rm w}={\gamma}^{\rm h}\,.
\end{equation}
Here we prepare a representation of the wave nest relevant for the
future use.

In the operator form, the wave equation (\ref{Eq 1}) is
$u_{tt}+L_0^*u=0$\,\,in \,$\mathscr H$. The control $f$ from (\ref{Eq 3})
determines the harmonic (i.e., ${\mathscr K}$-valued) function $\check f$
that depends on time as a parameter and satisfies
$$
\Delta\check f(\cdot,t)=0\quad{\rm
in}\,\,{\Omega}\setminus{\Gamma};\qquad \check
f(\cdot,t)|_{\Gamma}=f(\cdot,t),\,\,t\geqslant 0.
$$
Representing $u=\check f+w$, we have $w_{tt}+Lw=-\check
f_{tt}$\,\,\, (recall (\ref{Eq oper triple}))\,\, and, integrating
the latter equation, get
$$
u^f(\cdot,t)=\check
f(\cdot,t)-L^{-\frac{1}{2}}\int_0^t\sin\,[(t-s)L^{\frac{1}{2}}]\,\check
f_{tt}(\cdot,s)\,ds\quad {\rm in}\,\,\,\mathscr H\,,\quad \check
f(\cdot,t)\in{\mathscr K} .
$$
As the result, we arrive at a representation in $\mathscr H$;
\begin{align}
\notag &{\gamma}^{\rm w}=\{\overline{{\mathscr U}^t}\}_{t\geqslant
0}\,;\quad {\mathscr U}^t=\{u^h(\cdot,t)\,|\,\,h\in
C^\infty([0,t];{\mathscr K}),\,\,{\rm
supp\,}h\subset(0,t]\},\\
& \label{Eq repres gamma^w}
u^h(\cdot,t)=h(\cdot,t)-L^{-\frac{1}{2}}\int_0^t\sin\,[(t-s)L^{\frac{1}{2}}]\,
h_{tt}(\cdot,s)\,ds,\quad t\geqslant 0
\end{align}
in terms of the $\mathscr H$-  and ${\mathscr K}$-valued
functions. By the latter, to describe the wave nest image
${\tilde{\gamma}}^{\rm w}=U{\gamma}^{\rm w}$ in a space ${\tilde
{\mathscr H}}=U\mathscr H$ it suffices just to provide the proper
terms of (\ref{Eq repres gamma^w}) with tildes:
\begin{align}
\notag &{\tilde{\gamma}}^{\rm w}=\{\overline{{\tilde{\mathscr
U}}^t}\}_{t\geqslant 0}\,;\quad \tilde{\mathscr U}^t=\{\tilde
u^h(t)\,|\,\,h\in C^\infty([0,t];\tilde{\mathscr K}),\,\,{\rm
supp\,}h\subset(0,t]\},\\
& \label{Eq repres tilde gamma^w} {\tilde u}^h(t)=h(t)-\tilde
L^{-\frac{1}{2}}\int_0^t\sin\,[(t-s)\tilde L^{\frac{1}{2}}]\,
h_{tt}(s)\,ds\quad {\rm in}\,\,\tilde{\mathscr H},\qquad
t\geqslant 0.
\end{align}

\subsubsection*{Wave isotony}
{\noindent$\bullet$\,\,\,} Let us consider one more dynamical system
\begin{align}
\label{Eq 11}& v_{tt}-\Delta v = \psi  && {\rm in}\,\,\,({\Omega}\setminus{\Gamma})\times{\mathbb R_+};\\
\label{Eq 21}& v|_{t=0}=v_t|_{t=0}=0 && {\rm in}\,\,{\Omega};\\
\label{Eq 31}& v = 0 && {\rm on}\,\,\,{\Gamma}\times \overline{\mathbb R_+};
\end{align}
with a {\it domain control} $\psi=\psi(x,t)$; by $v=v^\psi(x,t)$
we denote a solution. For a smooth $\psi$ vanishing near $t=0$ the
solution $v^\psi$ is unique and classical. In the equivalent form,
regarding $\psi$ and $v^\psi$ as the $\mathscr H$-valued functions of
time, we have the system
\begin{align}
\label{Eq 1+}& v_{tt}+L v = \psi  && {\rm in}\,\,\mathscr H,\,\,\,t>0;\\
\label{Eq 2+}& v|_{t=0}=v_t|_{t=0}=0 && {\rm in}\,\,\mathscr H;
\end{align}
and represent its solution in the well-known form
\begin{equation}\label{Eq v^psi}
v^\psi(t)\,=\,L^{-\frac{1}{2}}\int_0^t\sin\,[(t-s)L^{\frac{1}{2}}]\,
\psi(s)\,ds\qquad {\rm in}\,\,\mathscr H,\,\,\,t\geqslant0.
\end{equation}
The right hand side is well defined for $\psi\in L_2([0,t];\mathscr H)$
and is referred to as a generalized solution to (\ref{Eq
11})--(\ref{Eq 2+}).
\smallskip

{\noindent$\bullet$\,\,\,} We say that a control $\psi$ acts from a subspace $\mathscr
A\subset\mathscr H$ if $\psi(t)\in\mathscr A$ holds for all $t\geqslant 0$. The
set
\begin{equation}\label{Eq set V}
{\mathscr V}^t_\mathscr A\,:=\,\{v^\psi(t)\,|\,\,\psi\in L_2([0,t];\mathscr A)\},\qquad
t> 0
\end{equation}
is called reachable from $\mathscr A$ at the moment $t$.

An analog of the property (\ref{Eq Holmgren}) is a {\it  domain
controllability}, which is also derived from the
Holmgren-John-Tataru Theorem.
\begin{Proposition}\label{Prop Holmgren domain}
If $G\subset{\Omega}$ is an open subset provided $\mu(\partial G)=0$
then the equality
\begin{equation}\label{Eq Holmgren domain}
\overline{{\mathscr V}^t_{\mathscr H\langle G\rangle}}\,=\,\mathscr H\langle G^t\rangle,\qquad t>0
\end{equation}
holds.
\end{Proposition}

{\noindent$\bullet$\,\,\,} Recall that ${\mathfrak L}$ is the Hilbert subspace lattice in $\mathscr H$.
The map $\mathscr A\mapsto\overline{{\mathscr V}^t_\mathscr A}$ is evidently isotonic
w.r.t. $\mathscr A$ and monotone w.r.t. $t$. Hence, we have a {\it wave
isotony}
\begin{equation}\label{Eq wave isotony on L} I_{\rm
w}\,=\,\{i^t_{\rm w}\}_{t\geqslant 0}\,, \quad i^t_{\rm
w}:\mathscr A\mapsto\overline{{\mathscr V}^t_\mathscr A}
\end{equation}
defined on the {\it whole} lattice ${\mathfrak L}$.

The assumptions of Proposition \ref{Prop Holmgren domain} mean
that $G\in\langle G\rangle \in {\mathfrak L}_{\Omega}^{\rm h}$ holds. Comparing
(\ref{Eq wave isotony on L}) with (\ref{Eq Hilbert isotony}) and
taking into account (\ref{Eq Holmgren domain}), we conclude that
Proposition \ref{Prop Holmgren domain} implies
\begin{equation}\label{Eq I^h=I^w}
I_*\,=\,I_{\rm w}{\upharpoonright}{\mathfrak L}_{\Omega}^{\rm h}\,.
\end{equation}

By (\ref{Eq I^h=I^w}) and with regard to
 (\ref{Eq
gamma^w=gamma^h}),
 we get the following scheme of determination of
a simple manifold:
\begin{equation}\label{Eq scheme wave}
{\gamma}^{\rm w},\,I_{\rm
w}\,\,\overset{\sqcup}{\,\,\Rightarrow}\,\,\mathfrak L^{\rm
h}_{\Gamma}\,\,\Rightarrow\,\,{\rm At\,}\overline{I_{\rm
w}\mathfrak L^{\rm h}_{\Gamma}}\,\,\Rightarrow\,\,({\Omega}_*,\rm
d_*).
\end{equation}

{\noindent$\bullet$\,\,\,} Let a space $\tilde {\mathscr H}=U\mathscr H$ and operator $\tilde {L}=ULU^*$
be given. Then one can determine  the isotony $\tilde I_{\rm
w}=\Phi I_{\rm w}\Phi^*$ in $\tilde {\mathscr H}$ as follows:
\begin{align}
\notag & \tilde I_{\rm w}=\{\tilde i_{\rm w}^t\}_{t\geqslant
0}\,,\,\,\,\tilde
i_{\rm w}^t: \mathscr B \mapsto \tilde {\mathscr V}_\mathscr B^t\,\,\,\,(\mathscr B\subset\tilde{\mathscr H}),
\quad \tilde{ {\mathscr V}}^t_\mathscr B\,\overset{(\ref{Eq set V})}=\,\{\tilde v^\psi(t)\,|\,\,\psi\in L_2([0,t];\mathscr B)\},\\
\label{Eq repres tilde I}& \tilde v^\psi(t)\,\overset{(\ref{Eq
v^psi})}=\,{\tilde L}^{-\frac{1}{2}}\int_0^t\sin\,[(t-s){\tilde
L}^{\frac{1}{2}}]\,
\psi(s)\,ds\,.
\end{align}

At last, if we have an $\tilde{\mathscr H}=U\mathscr H$ along with
the nest $\tilde{\gamma}^{\rm w}=U{\gamma}^{\rm w}$ and isotony
$\tilde I_{\rm w}=U I_{\rm w}U^*$ then the determination can be
fulfilled just by copying the scheme (\ref{Eq scheme wave}):
\begin{equation}\label{Eq scheme tilde wave}
\tilde{\gamma}^{\rm w},\,\tilde I_{\rm
w}\,\,\overset{\sqcup}{\,\,\Rightarrow}\,\,\tilde{\mathfrak
L}^{\rm h}_{\Gamma}\,\,\Rightarrow\,\,{\rm At\,}\overline{\tilde
I_{\rm w}\tilde{\mathfrak L}^{\rm
h}_{\Gamma}}\,\,\Rightarrow\,\,(\tilde{\Omega}_*,\tilde{\rm
d}_*)\overset{\rm isom}=({\Omega},\rm d).
\end{equation}

\section{Determination of manifold}\label{Sec Reconstr}

\subsubsection*{Spectral representation}
The Dirichlet spectral problem in ${\Omega}$ is
$$
\begin{cases}
-\Delta\phi=\lambda\phi &{\rm in}\,\,\,{\Omega}\setminus{\Gamma};\cr
\phi=0 & {\rm on}\,\,\,{\Gamma};
\end{cases}
$$
or, equivalently, $L\phi=\lambda\phi$ \,in $\mathscr H$. Let
$0<\lambda_1<\lambda_2\leqslant\dots$ be the spectrum of the problem (of
operator $L$) and $\phi_1,\phi_2,\dots$ the corresponding
eigenfunctions normalized by $(\phi_k,\phi_l)=\delta_{kl}$.

Let $\Phi:\mathscr H\to {\bf l}_2=:\tilde {\mathscr H}$,
$$
\Phi y\,=\,\tilde y:=\{c_k\}_{k=1}^\infty,\quad c_k=(y,\phi_k)
$$
be the Fourier transform that diagonalizes  operator $L$, i.e.,
represents it in the form $\tilde L:=\Phi L\Phi^*={\rm
diag\,}\{\lambda_1,\lambda_2,\dots\}$.

The subspace $\tilde{\mathscr K}=\Phi{\mathscr K}$ is referred to as the harmonic
function subspace in the spectral representation.

\subsubsection*{Completing the determination}

At last, assume that the spectrum $\lambda_1,\lambda_2,\dots$ and
the subspace $\tilde {\mathscr K}\subset {\bf l}_2$ are given. Let
$({\Omega},\rm d)$ be a simple manifold. Then it (more precisely,
its isometric copy) can be constructed as follows.
\smallskip

\noindent{\it Step 1\,\,\,}Putting $U=\Phi$, determine the
boundary nest $\tilde{\gamma}^{\rm w}$ by (\ref{Eq repres tilde
gamma^w}).
\smallskip

\noindent{\it Step 2\,\,\,}Determine the isotony $\tilde I_{\rm
w}$ by (\ref{Eq repres tilde I}).
\smallskip

\noindent{\it Step 3\,\,\,}Determine the Hilbert wave copy
$(\tilde{\Omega}_*,\tilde{\rm d}_*)$ by the scheme (\ref{Eq scheme
tilde wave}).

\subsubsection*{Commments}
{\noindent$\bullet$\,\,\,} A way to give the spectral image $\tilde {\mathscr K}$ of the harmonic
subspace ${\mathscr K}$ is to chose a basis $h_1,h_2,\dots$ in ${\mathscr K}$, to
form the matrix $\varkappa:=\{(h_i,\phi_k)\}_{i,k\geqslant 1}$ and then
determine $(\tilde{\Omega}_*,\tilde{\rm d}_*)$ from
$\lambda_1,\lambda_2,\dots$ and $\varkappa$. Such a trick does not
assume the boundary ${\Gamma}$ to be known and, hence, is admissible
in the M.Kac problem. Moreover, it is motivated by the following
observation.

Consider the classical inverse problem on determination of the
inhomogeneous string density $\rho>0$ from the {\it spectral data}
$\{\lambda_k;\nu_k\}_{k\geqslant 1}$. The data consist of the spectrum
$0<\lambda_1<\lambda_2<\dots$ of the problem
$$
\begin{cases}
-\rho^{-1}\phi''=\lambda\phi &{\rm in}\,\,\,(0,1);\cr
\phi(0)=\phi(l)=0
\end{cases}
$$
and the constants $\nu_k=\phi'_k(0)$, where $\phi_k$ are the
eigenfunctions normalized by $(\phi_i,\phi_k)=\delta_{ik}$ in
$\mathscr H=L_{2,\,\rho}(0,l)$. With the problem one associates the
operators
\begin{align*}
& L_0=-\rho^{-1}(\,\cdot\,)''{\upharpoonright} \{y\in
H^2[0,1]\,|\,\,y(0)=y'(0)=y(1)=0\},\\
& L= -\rho^{-1}(\,\cdot\,)''{\upharpoonright} \{y\in
H^2[0,1]\,|\,\,y(0)=y(1)=0\},\\
& L_0^*=-\rho^{-1}(\,\cdot\,)''{\upharpoonright} \{y\in
H^2[0,1]\,|\,\,y(1)=0\},
\end{align*}
which constitute the triple $L_0\subset L\subset L_0^*$ (see
(\ref{Eq oper triple})). The relevant 'harmonic function subspace'
is ${\mathscr K}={\rm Ker\,}L_0^*={\rm
span\,}\{h\},\,\,h(x):=1-x$;\,\,${\rm dim\,}{\mathscr K}=1$. In the mean
time, we have
\begin{align*}
&
(h,\phi_k)=\int_0^1(1-x)\,\phi_k(x)\,\rho(x)\,dx=-\frac{1}{\lambda_k}\int_0^1(1-x)\,\phi''_k(x)\,dx=\frac{\phi'_k(0)}{\lambda_k}.
\end{align*}
Therefore $\nu_k=\lambda_k(h,\phi_k)$, so that to give
$\{\nu_k\}_{k\geqslant 1}$ is to provide the row
$\varkappa=\{(h,\phi_k)\}_{k\geqslant 1}$.
\smallskip

Thus, in the multidimensional case, it looks reasonable to regard
$\left\{\lambda_k; \{\varkappa_{ik}\}_{i\geqslant 1}\right\}_{k\geqslant 1}$
with $\varkappa_{ik}=(h_i,\phi_k)$ as a version of the spectral
data relevant for the M.Kac problem.
\smallskip

{\noindent$\bullet$\,\,\,} The isospectral drums are the ones with the same spectrum
$\sigma(L)=\{\lambda_k\}_{k\geqslant 1}$. As the above considerations
show, such drums differ by a position of the harmonic subspace
${\mathscr K}$ w.r.t. the eigenbasis of the operator $L$. Therefore, any
information on such a position may be useful for determination.
For instance, some information is contained in the {\it angular
spectrum} of ${\mathscr K}$ that is the sequence $\{\alpha_k\}_{k\geqslant
1},\,\,\alpha_k:=\|P_{\mathscr K}\phi_k\|$, where $P_{\mathscr K}$ projects in $\mathscr H$
onto ${\mathscr K}$. It does not fix ${\mathscr K}$ uniquely but reduces the reserve
of ${\mathscr K}$'s corresponding to the given spectrum.
\smallskip

{\noindent$\bullet$\,\,\,} Assume that we are given the spectral representation $\tilde
L_0=\Phi L_0\Phi^*$ of the minimal Laplacaian (see (\ref{Eq oper
triple})). Then, along with it, we have the operators $\tilde L$
(as the Friedrichs extension of $\tilde L_0$) and $\tilde L_0^*$
(as adjoint). Further, we find $\sigma(\tilde L)=\sigma(L)$ and
$\tilde{\mathscr K}={\rm Ker\,}\tilde L_0^*$. As the result, we get the
data $\sigma(L),\tilde{\mathscr K}$, which determine ${\Omega}$. Sure, to
repeat such a trick it suffices to have $L_0$ in {\it any}
representation. By the latter, our technique enables to determine
manifolds from the characteristic operator function of $L_0$ (see,
e.g., \cite{Strauss}) or from any kind of the data that determine
$L_0$ up to the unitary equivalence.

\subsubsection*{A hypothesis}
{\noindent$\bullet$\,\,\,} Let $A_0\subset A_0^*$ be a symmetric operator in $\mathscr
H$ provided $\overline{{\rm Dom\,}A_0}=\mathscr H$ and $(A_0y,y)\geqslant
\nu\,\|y\|^2\,\,\,(\nu>0)$. Operator $A_0$ determines its
(extremal) self-adjoint extensions $A_M$ and $A_\mu$ by M.Krein
and Friedrichs respectively (see, e.g., \cite{BirSol,MMM}). The
following operator theory fact is relevant to our approach to
M.Kac's problem.
\begin{Proposition}\label{Prop Krein ext}
The pair $A_M$,$A_\mu$ determines the operator $A_0$.
\end{Proposition}
\begin{proof}
Recall the well-known facts and relations:
\begin{align}
\notag & {\rm Dom\,}A_M={\rm Dom\,}A_0\dotplus \mathscr
K,\quad\mathscr
K:={\rm Ker\,}A_0^*= {\rm Ker\,}A_M;\\
\notag & {\rm Dom\,}A_\mu={\rm Dom\,}A_0\dotplus
A_\mu^{-1}\mathscr K.
\end{align}
By ${\rm Ran\,}A_0=\mathscr H\ominus {\rm Ker\,}A_0^*=\mathscr
H\ominus {\rm Ker\,}A_M=\mathscr H\ominus{\mathscr K}$, the
relation $A_0\subset A_\mu$ implies
$$
{\rm Dom\,}A_0=A_0^{-1}{\rm Ran\,}A_0=A_0^{-1}(\mathscr H\ominus
{\rm
 Ker\,}A_M)=A_\mu^{-1}(\mathscr H\ominus{\mathscr K})
$$
($A_0^{-1}$ is the inverse on the image) that leads to
\begin{equation}\label{Eq A0 via A mu and K}
 A_0=A_0\upharpoonright {\rm Dom\,}A_0=A_\mu\upharpoonright {\rm Dom\,}A_0=A_\mu\upharpoonright{A_\mu^{-1}(\mathscr H\ominus {\mathscr K})}.
\end{equation}
\end{proof}
The representation (\ref{Eq A0 via A mu and K}) implies the
following.
\begin{Corollary}\label{Corol}
Operator $A_0$ (and, hence, $A_0^*$) is determined by the pair
$A_\mu,{\mathscr K}$.
\end{Corollary}
\noindent  Analyzing the scheme (\ref{Eq scheme tilde wave}), one
can easily recognize: it is the latter fact that, being applied to
the triple (\ref{Eq oper triple}), provides the determination.

\smallskip

{\noindent$\bullet$\,\,\,} Here we formulate a hypothesis.

Let $A_0$ and $A'_0$ be the densely defined closed symmetric
positive definite operators in $\mathscr H$. Let $A_M,\, A_\mu$
and $A'_M,\,A'_\mu$  be their extremal extensions.
\begin{Hypothesis}\label{Hyp}
Let $VA_\mu=A'_\mu V$ and $WA_M=A'_MW$ holds, where $V,W$ are the
unitary operators. Then one has $V=W=:U$ and $UA_0=A'_0U$.
\end{Hypothesis}

\noindent If the hypothesis is true and its conditions are
satisfied, then all positive extensions of $A_0$ and $A_0'$ are
intertwined by a single unitary operator $U$. One more consequence
is that the relations $A_\mu=A_\mu'$ and $UA_M=A'_MU$ imply
$U=\mathbb I$, leading to $A_M=A'_M$ and ${\mathscr K}={\mathscr
K}'$.
\medskip

{\noindent$\bullet$\,\,\,} Let the extensions $A_\mu$ and $A_M$ of $A_0$ have the
discrete spectra $\sigma(A_\mu)$ and $\sigma(A_M)$. The spectra
determine these extensions up to unitary equivalence. Now assume
that the Hypothesis \ref{Hyp} is valid, and let the extensions
$A'_\mu$ and $A'_M$ of $A'_0$ be such that
$\sigma(A_\mu)=\sigma(A'_\mu)$ and $\sigma(A_M)=\sigma(A'_M)$
holds. Then the operators $A_0$ and $A_0'$ turn out to be
unitarily equivalent, and we get $A_0'=UA_0U^*$ and ${\mathscr K}'=U{\mathscr K}$.
Applied to M.Kac problem, this would give a chance to determine a
simple manifold from two spectra $\sigma(L)$ and $\sigma(L_M)$,
where $L_M$ is the M.Krein extension of the minimal Laplacian
$L_0$ (see (\ref{Eq oper triple})). Note that such a situation
does occur in the one-dimensional case mentioned in Comments. In
this case, we have $L_M=-\rho^{-1}(\,\cdot\,)''{\upharpoonright} \{y\in
H^2[0,1]\,|\,\,y(0)+ly'(0)=0,\,\,y(1)=0\}$ and the density $\rho$
indeed can be determined from two spectra corresponding to
different boundary conditions at $x=0$.

In the multidimensional problem, to give $\sigma(L_M)$ as
additional (to $\sigma(L)$) data instead of the matrix $\varkappa$
would be more attractive.
\medskip

{\noindent$\bullet$\,\,\,}
 There are non-isometric drums, whose
spectra of the Dirichlet problem and the Neumann problem coincide.
\cite{Buser}. However, the extension of $L_N$ of the minimal
Laplacean $L_0$, which satisfies Neumann's condition
$\partial_ny\big|_{\Gamma}=0$, does not coincides with its soft
extension $L_M$ and, in contrast to the latter, is of no invariant
operator sense. We continue to hope for the Hypothesis \ref{Hyp}.

\bigskip

\noindent{\bf Key words:}\,\,\,M.Kac problem, augmented data,
lattice theory, dynamical system with boundary control.
\smallskip

\noindent{\bf MSC:}\,\,\,47Axx,\,\, 47B25,\,\, 35R30.


\begin{thebibliography}{99}

\bibitem{B Kac_Prob}
M.I.Belishev.
\newblock {On M.Kac problem on the domain shape determination from the Dirichlet spectrum.}
\newblock{\em Journal of Soviet Mathematics}, 06/1991, 55(3):1663-1672.
DOI:10.1007/BF01098204.

\bibitem{B Obzor IP 97}
M.I.Belishev.
\newblock {Boundary control in reconstruction of manifolds and
metrics (the BC method).}
\newblock {\em Inverse Problems}, 13(5): R1--R45, 1997.

\bibitem{B JOT}
M.I.Belishev.
\newblock {A unitary invariant of a semi-bounded operator in reconstruction
of manifolds.}
\newblock{\em Journal of Operator Theory}, Volume 69 (2013), Issue 2, 299-326.

\bibitem{BSim_FAN}
Belishev, M.I. and Simonov S.A.
\newblock {A Wave Model of Metric Spaces.}
\newblock {\em Functional Analysis and Its Applications}, April 2019, Volume 53,
Issue 2, pp 79–85. Print ISSN 0016-2663, Online ISSN 1573-8485,
https://doi.org/10.1134/S0016266319020011).

\bibitem{BSim_MS}
M.I.Belishev and S.A.Simonov.
\newblock {The wave model of a metric space
with measure and an application.}
\newblock {\em Sbornik: Mathematics}, 211:4,
521--538. DOI: https://doi.org/10.1070/SM9242

\bibitem{Birhoff}
G.Birkhoff.
\newblock{Lattice theory.}
\newblock{\em Providence - Rhod Island}, 1967.

\bibitem{BirSol}
M.Sh.Birman, M.Z.Solomak.
\newblock{Spectral Theory of Self-Adjoint Operators in Hilbert Space.}
\newblock{\em Reidel Publishing Comp.}, 1987.

\bibitem{Buser}
P.Buser, J.Conway, K.D.Semmler.
\newblock{Some planar isospectral
domains.}
\newblock{\em Int. Math. Res. Notices}, 9 (1994), 391–400.

\bibitem{MMM}
V.F.Derkach, M.M.Malamud.
\newblock {Theory of symmetric operator extensions
and boundary value problems.\,\,\,(in Russian)}
\newblock{\em Ki\"iv}, 2017. ISBN 966-02-2571, ISBN
978-966-02-8267-4 (v.104)

\bibitem{Giraud}
O.Giraud, K.Shas.
\newblock{Hearing shapes of drums –
mathematical and physical aspects of isospectrality.}
\newblock{\em Reviews of Modern Physics}, 82 (2010), 2213–2255.

\bibitem{Strauss}
A.V.Strauss
\newblock {Functional models and generalized spectral functions of symmetric operators.}
\newblock{\em Saint-Petersburg Mathematical Journal}, 10:5 (1999), 733–784.

\end{thebibliography}
\end{document}